\documentclass[journal]{IEEEtran}

\ifCLASSINFOpdf
\else
   \usepackage[dvips]{graphicx}
\fi
\usepackage{url}

\usepackage{graphicx}
\usepackage{cite}
\usepackage{amsmath,amsfonts}
\usepackage{algorithmic}
\usepackage{graphicx}
\usepackage{textcomp}
\usepackage{xcolor}
\usepackage{amsthm}
\usepackage{caption}
\usepackage{relsize}
\usepackage{subcaption}
\usepackage{float}
\usepackage{bbold}
\def\BibTeX{{\rm B\kern-.05em{\sc i\kern-.025em b}\kern-.08em
    T\kern-.1667em\lower.7ex\hbox{E}\kern-.125emX}}
    
\usepackage{systeme}
\usepackage{color}
\usepackage{pgfplots}
\usepackage{bbm}
    
  \newcommand{\figref}[1]{Fig.~\protect\ref{#1}}

                            \newcommand{\pto}{\overset{P}\longrightarrow }

\long\def\comment#1{}

\DeclareMathOperator*{\argmin}{arg\,min}

\newfont{\bbb}{msbm10 scaled 700}

\newfont{\bb}{msbm10 scaled 1100}

\newcommand{\av}{{\boldsymbol a}}
\newcommand{\bv}{{\boldsymbol b}}

\newcommand{\gv}{{\boldsymbol g}}

\newcommand{\rv}{{\boldsymbol r}}
\newcommand{\sv}{{\boldsymbol s}}

\newcommand{\uv}{{\boldsymbol u}}
\newcommand{\wv}{{\boldsymbol w}}
\newcommand{\vv}{{\boldsymbol v}}
\newcommand{\xv}{{\boldsymbol x}}

\newcommand{\Am}{{\boldsymbol A}}

\newcommand{\Cm}{{\boldsymbol C}}

\newcommand{\Hm}{{\boldsymbol H}}

\renewcommand{\arg}{{\hbox{arg}}}

\newtheorem{theorem}{Theorem}

\newtheorem{remark}{Remark}

\begin{document}

\title{Performance Analysis of Regularized Convex Relaxation for Complex-Valued Data Detection}

\author{Ayed M. Alrashdi, and Houssem Sifaou
\thanks{A. M. Alrashdi is with the Department of Electrical Engineering, College of Engineering, University of Ha'il, P.O. Box 2440, Ha'il, 81441, Saudi Arabia (e-mail: am.alrashdi@uoh.edu.sa).}
\thanks{H. Sifaou is with Department of Electrical and Electronic Engineering, Imperial College London, London SW7 2AZ, U.K. (e-mail: h.sifaou@imperial.ac.uk).}
}

\maketitle

\begin{abstract}
In this work, we study complex-valued data detection performance in massive multiple-input multiple-output (MIMO) systems. We focus on the problem of recovering an $n$-dimensional signal whose entries are drawn from an arbitrary constellation $\mathcal{K} \subset \mathbb{C}$ from $m$ noisy linear measurements, with an independent and identically distributed (i.i.d.) complex Gaussian channel.
Since the optimal maximum likelihood (ML) detector is computationally prohibitive for large dimensions,
many convex relaxation heuristic methods have been proposed to solve the detection problem.
In this paper, we consider a regularized version of this convex relaxation that we call the regularized convex relaxation (RCR) detector and sharply derive asymptotic expressions for its mean square error and symbol error probability. 
Monte-Carlo simulations are provided to validate the derived analytical results.
\end{abstract}

\IEEEpeerreviewmaketitle

\section{Introduction}
 \label{sec:introd}
%
\IEEEPARstart{D}{etection} of complex-valued data from noisy linear measurements appears often in many communication applications, such as massive multiple-input multiple-output (MIMO) signal detection \cite{ngo2013energy}, multiuser detection \cite{verdu1998multiuser}, etc.
In this work, we consider the problem of recovering an $n$-dimensional complex-valued signal whose entries are drawn from an arbitrary constellation $\mathcal{K} \subset \mathbb{C}$ from $m$ noisy linear measurements in a massive MIMO application.
Although the maximum likelihood (ML) detector can achieve excellent performance, its computational complexity becomes prohibitive as the problem size increases \cite{verdu1998multiuser}.  To achieve an acceptable performance with a low computing complexity,  various convex optimization-based heuristics have been developed.
One such heuristic method is to relax the discrete set $\mathcal{K}$ to a convex and continuous set $\mathcal{V}$, and then solve the detection problem using a regularized convex optimization program followed by hard-thresholding. We call this method the \emph{regularized convex relaxation} (RCR) detector, and we sharply analyze its asymptotic performance in terms of its mean square error (MSE) and symbol error probability (SEP) in the limit of $m, n \to \infty$ with a fixed rate. We assume an independent and identically distributed (i.i.d.) \textit{complex} Gaussian channel matrix and additive white Gaussian noise. 
The technical tool used in the analysis is the convex Gaussian min-max theorem (CGMT) framework \cite{thrampoulidis2016precise, stojnic2013framework}.

The CGMT has been used to analyze the error performance of many optimization problems \cite{thrampoulidis2016precise, thrampoulidis2018symbol, thrampoulidis2015regularized, atitallah2017box, atitallah2017ber, alrashdi2020optimum, alrashdi2019precise}. However, prior performance analysis was only established for real-valued constellations such as BPSK and PAM using the box-relaxation method \cite{thrampoulidis2018symbol}.
For the complex-valued constellations,  to the best of our knowledge, the CGMT was only applied in one work \cite{abbasi2019performance} which considered convex-relaxation performance, but without regularization. 

Our main contribution is to extend these performance results to complex-valued constellations with regularized detectors,  and show the additional performance gains attained by adding the convex relaxation constraint. 
Our results also enable us to select the optimal regularization factor to further improve the detection performance.
As a concrete example, we focus our attention on studying the performance of the RCR detector for phase-shift keying (PSK) and quadrature amplitude modulation (QAM) constellations. 
\section{Problem Formulation}
\subsection{Notation}
The basic notations used throughout this article are gathered here.
For a complex scalar $z \in \mathbb{C}$, $z_{\mathsmaller{R}}$, and $z_{\mathsmaller{I}}$ represent the real and imaginary parts of $z$, respectively and $|z| = \sqrt{z_{\mathsmaller{R}}^2+z_{\mathsmaller{I}}^2}$.  We use the letter $j$ to denote the complex unit, i.e., $j^2 = -1$.  A real Gaussian distribution with mean $\mu$ and variance $\sigma^2$ is denoted by $\mathcal{N}(\mu,\sigma^2)$.
Similarly,  $\mathcal{CN}(\mu,\sigma^2)$ denotes a complex Gaussian distribution with real
and imaginary parts drawn independently from $\mathcal{N}(\mu_{\mathsmaller{R}}, \sigma^2/2)$ and $\mathcal{N}(\mu_{\mathsmaller{I}}, \sigma^2/2)$ respectively. $X \sim p_X$ implies that the random variable $X$ has a density $p_X$. 
Bold lower-case letters are reserved for vectors, e.g., $\xv$, with $x_i$ denoting its $i$-th entry. 
The Euclidean norm of a vector is denoted by $\|\cdot\|$. 
Matrices are represented by bold upper-case letters, e.g., $\Am$, while $(\cdot)^{\top}$ represents the transpose operator.
We reserve the letters $G$ and $Z$ to denote real standard normal random variables. Similarly, $G_{c}$ is reserved
to denote a complex $\mathcal{CN}(0, 2)$ normal random variable. $\mathbb{E}[\cdot]$ and $\mathbb{P}[\cdot]$ denote the expectation and probability operators, respectively.
Notations ``${\overset{d} =}$" and ``$\pto$" are used to denote equivalence in distribution, and convergence in probability, respectively.
Finally, for a closed and nonempty convex set $\mathcal{V} \subset \mathbb{C}$, and for any vector $\xv \in \mathbb{C}^n$, we define its \textit{distance} and \textit{projection} functions, respectively, as follows
\begin{align}
{\mathcal{D}}(\xv;\mathcal{V}) = \min_{\av \in \mathcal{V}^n} \| \xv - \av  \|,
\end{align}
\begin{align}
{\mathbf{\Pi}}(\xv;\mathcal{V}) = \arg \min_{\av \in \mathcal{V}^n} \| \xv  - \av \|.
\end{align}
\subsection{Problem Setup}\label{sec:setup}
We need to recover an $n$-dimensional complex-valued transmit vector $\sv_{\rm{0}} \in \mathcal{K}^n \subset \mathbb{C}^n$, where $\mathcal{K}$ is the discrete transmit modulation constellation (e.g., PSK, QAM, etc.). The received signal vector $\rv \in \mathbb{C}^m$ is given by
\begin{align}
\rv =\Hm \sv_0 +\vv,
\end{align}
where $\Hm \in \mathbb{C}^{m \times n}$ is the MIMO channel matrix that has $\mathcal{C N} (0,\frac{1}{n})$ i.i.d. entries and $\vv \in \mathbb{C}^m$ is the noise vector with $\mathcal{C N} (0,\sigma^2)$ i.i.d. entries.
Under the current setup, the signal-to-noise ratio (SNR) is $\mathrm{SNR}$ $=1/\sigma^2$.\\
\textbf{Detector}: The optimum detector that minimizes the probability of error is the maximum likelihood (ML) detector which solves
\begin{align}
\widehat{\sv}_{\rm{\tiny{ML}}}: = {\rm{arg}} \min_{\sv \in \mathcal{K}^n} \frac{1}{2}\| \Hm \sv -  \rv \|^2.
\end{align}
This is computationally prohibitive in our massive MIMO setup; due to the discreteness nature of the constraints set $\mathcal{K}$. Instead, in this paper, we consider the \textbf{regularized convex relaxation (RCR)} detector. The RCR recovers $\sv$ following the next two steps:
\begin{subequations}\label{eq:BRO}
\begin{align}
& \widehat{\sv}: = \arg \min_{\sv \in \mathcal{V}^n} \frac{1}{2}\|  \Hm \sv -  \rv  \|^2 + \frac{\zeta}{2} \| \sv \|^2,\label{a}  \\
& s_i^{\star}: = \arg \min_{c \in \mathcal{K}} | c - \widehat{s}_i|,\label{b} \quad \quad i =1,2,\cdots,n,
\end{align}
\end{subequations}
where in the first step \eqref{a}, the discrete set $\mathcal{K}$ is relaxed to a \textit{convex} set $\mathcal{V}$, and then we solve a regularized version of this relaxed problem with $\zeta \geq 0 $ being the regularization factor.
In the second step \eqref{b}, each entry of $\widehat{\sv}$ is mapped to its closest point in $\mathcal{K}$ to produce the final estimate $\sv^{\star}$.
%
%
\subsection{Performance Measures}
In this paper, we provide sharp performance analysis of the RCR detector as a function of the problem parameters such as $\rm SNR$, $m,n$, $\mathcal{K}$ and $\mathcal{V}$. We will consider two different performance measures, namely the MSE and the SEP discussed next.

\noindent
{\bf{Mean Square Error (MSE)}}: This measures the performance of the \textit{estimation} step of the detector (first step in \eqref{a}) and is defined as:
\begin{align}
{\rm{MSE}} := \frac{1}{n} \| \sv_0 -\widehat{\sv} \|^2.
\end{align}
Another important performance metric is the symbol error probability.\\
\textbf{Symbol Error Probability (SEP)}: the symbol error rate (SER) characterizes the performance of the \textit{detection} process (second step \eqref{b}) and is defined as:
\begin{equation}
{\rm{SER}} := \frac{1}{n} \sum_{i=1}^{n} {\mathbb{1}}_{\{{s}^\star_i \neq s_{0,i} \}},
\end{equation}
where ${\mathbb{1}}_{\{\cdot\}}$ represents the indicator function.\\
In relation to the SER is the symbol error probability (SEP) which is defined as the expectation of the SER averaged over the noise, the channel and the constellation. Formally, the symbol error probability denoted by ${\rm SEP}$ is given by:
\begin{equation}
	{\rm SEP}:=\mathbb{E}[{\rm SER}] =\frac{1}{n}\sum_{i=1}^n \mathbb{P}\left[{s}^\star_i \neq s_{0,i} \right].
	\label{eq:SEP}
\end{equation}
Next, we introduce the notation $\mathcal{V}_x$ for $x \in \mathcal{K}$, as the set of all points in $\mathcal{V}$ that will be mapped to $x$ in \eqref{b}. Equivalently
\begin{equation}
\mathcal{V}_x : = \big\{ b \in \mathcal{V}: \forall a \in \mathcal{K}, | b-x|< |b-a| \big\}.
\end{equation}
With this notation at hand, we can rewrite the $\rm SEP$ in \eqref{eq:SEP} as
\begin{equation}
{\rm SEP} = \frac{1}{n}\sum_{i=1}^n \mathbb{P}\left[ \widehat{s}_i \notin \mathcal{V}_{s_{0,i}} \right],
\label{eq:SEP2}
\end{equation}
where $\widehat{s}_i$ is a minimizer of \eqref{a}.
\subsection{Technical Assumptions}\label{sec:assumptions}
We assume that the entries of $\sv_0$ are sampled i.i.d. from a density $p_{s_0}$, with $\mathbb{P}[b_1+j b_2] = \mathbb{P}[b_2+j b_1] \quad \forall b_1, b_2 \in \mathbb{R}$. Furthermore, we assume that $\sv_0$ is normalized to have zero-mean and unit-variance, i.e., $\mathbb{E}[S_0^2] =1$. The convex set $\mathcal{V}$ is assumed to be symmetric, i.e., if $(b_1+j b_2) \in \mathcal{V}$, then $(b_2+j b_1) \in \mathcal{V}$ as well.
Finally, we assume a high-dimensional regime in which $m, n \to \infty$ with a proportional rate $\kappa:=\frac{m}{n} \in (0,\infty)$.
\section{Asymptotic Performance Analysis}
\label{sec:Results}
\subsection{Main Results}
This section provides our main results on the performance evaluation of the RCR detector in the considered high dimensional setting. 
%
\begin{theorem}[Asymptotics of the RCR] \label{Theorem:BER}
Let $\rm MSE$, and $\rm SEP$ be the mean square error and symbol error probability of the RCR detector in \eqref{eq:BRO}, respectively, for an unknown signal $\sv_0 \in \mathcal{K}^n$ with entries sampled i.i.d.  from a distribution $p_{s_0}$. Let $\mathcal{V}$ be a convex relaxation of $\mathcal{K}$ that satisfies the assumption of Section \ref{sec:assumptions}. For fixed $\zeta\geq 0$, and $\kappa>0$, if following optimization problem
\begin{align}\label{min-max:eq}
\min_{\alpha>0} \max_{\beta>0}& \ {\kappa \alpha \beta} + \frac{\sigma^2 \beta}{2 \alpha} -\! \frac{\beta^2}{2} \!-\! \frac{\alpha \beta^2}{ \beta +2 \zeta \alpha} \!+\!  \left( \!\frac{\beta}{2 \alpha}\! - \!\frac{\beta^2}{{2 \alpha\beta}+4\zeta \alpha^2} \!\right)\nonumber \\
& \!+\! \left( \!\frac{\beta}{2 \alpha} \!+\!\zeta \!\right)\! \mathbb{E}\left[ {\mathcal{D}}^2\left(\frac{{\beta}}{\beta+2\zeta \alpha} \left(S_0 - \alpha G_c \right); \mathcal{V}\right) \right]
\end{align}
has a unique solution $(\alpha_*, \beta_*)$, then it holds in probability that
\begin{equation}\label{eq:mse}
\lim_{n \to \infty}{\rm MSE}= 2 \kappa \alpha_*^2 - \sigma^2,
\end{equation}
and
\begin{equation}\label{sep3}
\lim_{n \to \infty}{\rm SEP} \!=\! \mathbb{P}\! \left[ \!{\mathbf \Pi} \!\left( \! \frac{{\beta_*}}{\beta_*+2\zeta \alpha_*} \!\left(S_0 \! - \! \alpha_* G_c \!\right);\! \mathcal{V} \!\right) \! \notin \! \mathcal{V}_{S_0} \!\right]\!,
\end{equation}
where the expectation and probability in the above expressions are taken over $S_0 \sim p_{s_0}$ and $G_c \sim \mathcal{CN}(0,2)$.
\end{theorem}
\begin{proof}
See Section \ref{Proof}.
\end{proof}
Theorem \ref{Theorem:BER} provides high dimensional asymptotic expressions to calculate the MSE/SEP of the RCR detector under an arbitrary complex-valued constellation.
%
\begin{remark} The objective function in \eqref{min-max:eq} is convex-concave and only involves scalar variables, thus $\alpha_*$ and $\beta_*$ can be efficiently numerically calculated using the first-order optimality conditions. 
\end{remark}
\subsection{Modulation Schemes}
Using Theorem~1, we can sharply characterize the performance of the
RCR detector for a general complex-valued constellation $\mathcal{K}$, which can be relaxed to an arbitrary convex set $\mathcal{V}$.
To better understand our result and to show how to apply it to different schemes, we will focus on two conventional schemes; PSK and QAM constellations.
\subsubsection{$M$-PSK Constellation}
In an $M$-PSK constellation, where $M=2^k$, for some $k \in \mathbb{Z}_+$, each entry of $\sv_0$ is randomly drawn from the set $$\mathcal{K}=\left\{\exp\left({\frac{j2 \pi i}{M}} \right) , i =0,1,\cdots,M-1 \right\}, \quad M \geq 4.\footnote{This work focuses on complex-constellations, but the ``$M=2$"-case corresponds to the real-valued BPSK constellation, which has already been studied in \cite{atitallah2017ber}.}$$
The elements of $\mathcal{K}$ are uniformly distributed over the unit circle in the complex plane, and therefore we suggest the use of the so-called \emph{circular relaxation} (CR), where we choose the set $\mathcal{V}^{\rm CR} = \{ x \in \mathbb{C}: |x| \leq 1 \}$ as the convex relaxation set in \eqref{a}. The projection function on this set has the following form:
\begin{align*}
{\mathbf \Pi}(\xv;\mathcal{V}^{\rm CR}) = \begin{cases}
        \xv, \quad \text{if} \quad |\xv| \leq 1
        \\
      \frac{\xv}{|\xv|}, \quad \text{otherwise.}
        \end{cases}
\end{align*}
Due to the symmetric nature of the $M$-PSK constellation, the asymptotic $\rm{SEP}$ can be derived in the following closed-form:
\begin{align}
\lim_{n \to \infty} {\rm{SEP}} = \mathbb{P} \left[ \left| \frac{Z }  {  G - \frac{ 1}{\alpha_* } } \right|  \geq \tan \left( \frac{\pi}{M}\right)\right],
\end{align}
where $Z$ and $G$ are i.i.d. $\mathcal{N}(0,1)$ random variables.
\subsubsection{$M$-QAM Constellation}
In this paper, we will only consider square QAM modulation constellations, where $M = 2^{2k}$,
such as $16$-QAM, $64$-QAM, etc. Then, the constellation set is given by
\begin{align*}
\!\!\!\mathcal{K}\!=\!\left\{\!\! a\!+\! jb\!:\! a,b \! \in\! \left\{ \!\frac{\!-(\sqrt{M}-1)}{\sqrt{E_{\rm avg}}}\!,\!\frac{-(\sqrt{M}-3)}{\sqrt{E_{\rm avg}}}\!,\!\cdots\!, \!\!\frac{\sqrt{M}-1}{\sqrt{E_{\rm avg}}} \! \right\}\! \!\right\}\!,
\end{align*}
\noindent
where we normalize the constellation points by $E_{\rm avg} :=\frac{2(M-1)}{3}$; to have unit average power.\footnote{$E_{\rm avg}$ represents the average power of the non-normalized $M$-QAM symbols.}
The convex relaxation that is often used for this modulation is known as the \emph{Box-relaxation}
(BR) \cite{thrampoulidis2018symbol} which is given as
$$
\mathcal{V}^{\rm BR} = \left\{ (\alpha + j \beta) \in \mathbb{C} : |\alpha| \leq \frac{\sqrt{M}-1}{\sqrt{E_{\rm avg}}}, |\beta| \leq \frac{\sqrt{M}-1}{\sqrt{E_{\rm avg}}}\right\}.
$$
Similar to the previous section, In order to use Theorem \ref{Theorem:BER}, we need to form the projection and distance functions of $\mathcal{V}^{\rm BR}$ which is straightforward for a box set.
Then, $\rm SEP$ can be calculated using \eqref{sep3}.
Here, unlike the $M$-PSK case, the SEP of the recovery is not the same for different symbols in $\mathcal{K}$; since an $M$-QAM constellation has different types of points, namely inner, edge and corner points.
\section{Numerical Results}
In \figref{fig:performance_Reg}, we plotted the MSE and SEP performances as functions of the regularizer $\zeta$ for a $16$-QAM constellation with BR. 
These figures verify the accuracy of the prediction of Theorem~\ref{Theorem:BER} when compared to Monte-Carlo (MC) simulations. 
In addition, from those figures we can see a clear minimum value of $\zeta$ that gives the best MSE/SEP performance, thus Theorem~1 can be used to select the optimum value of the regularizer.

Furthermore, \figref{fig:PSK MSE/SEP} verifies the accuracy of the MSE/SEP predictions of Theorem~\ref{Theorem:BER} as functions of the SNR, for a $16$-PSK modulation scheme with circular relaxation. Note that although the theorem requires $m,n \to \infty$, it can be observed that the prediction is already accurate for $n = 128$.
In this figure, we also plotted the unconstrained regularized least-squares (RLS) (without convex relaxation, i.e., $\mathcal{V} = \mathbb{R}$). Besides, it can be seen that the RCR detector outperforms the RLS.
 
Finally, under the box relaxation, we apply Theorem~\ref{Theorem:BER} to sharply characterize the MSE and SEP of a $16$-QAM modulated system as functions of the SNR. The result is shown in \figref{fig:MSE/SER:QAM}, which again illustrates the high accuracy of our results, as well as that RCR detector outperforms the RLS.
%
\begin{figure}[ht]
\begin{subfigure}[h]{.5\textwidth}%
  \centering
\input{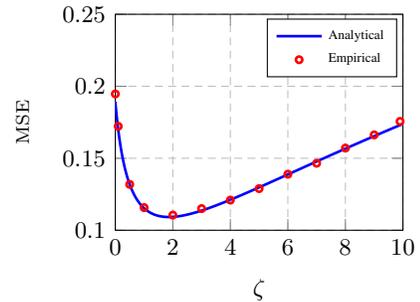}
\caption{\scriptsize{MSE performance vs. the regularizer.}}
\label{fig:mse_Reg}
\end{subfigure}
\begin{subfigure}[h]{.5\textwidth}%
  \centering
  \input{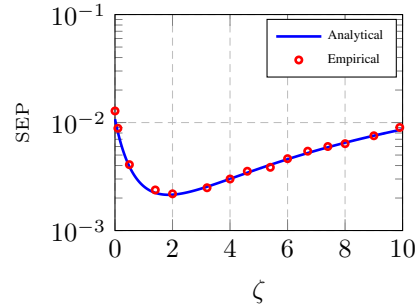}
\caption{\scriptsize{SEP performance vs. the regularizer.}}
\label{fig:ser_Reg}
\end{subfigure}
\caption{\scriptsize{MSE/SEP vs. the regularizer for $16$-QAM with BR, with $\kappa=1.5, n =128, \mathrm{SNR} = 15 \ \mathrm{dB}$. The data are averaged over $50$ independent MC trials.}}
\label{fig:performance_Reg}
\end{figure}
\begin{figure*}[ht]
\begin{subfigure}[h]{.5\textwidth}%
  \centering
%
%
\definecolor{mycolor1}{rgb}{0.00000,0.44700,0.74100}%
\definecolor{mycolor2}{rgb}{0.85000,0.32500,0.09800}%
\definecolor{OliveGreen}{rgb}{0,0.5,0}%
\begin{tikzpicture}[scale=1,font=\small]
    \renewcommand{\axisdefaulttryminticks}{4}
    \tikzstyle{every major grid}+=[style=densely dashed]
    \tikzstyle{every axis y label}+=[yshift=-10pt]
    \tikzstyle{every axis x label}+=[yshift=5pt]
    \tikzstyle{every axis legend}+=[cells={anchor=west},fill=white,
        at={(0.98,0.98)}, anchor=north east, font=\tiny ]
\begin{axis}[%
width=2.8in,
height=2.3in,
xmin=-5,
xmax=15,
ymin=0,
ymax=0.80,
      grid=major,
      scaled ticks=true,
       xlabel={\scriptsize{$\rm{SNR \ (dB)}$}},
       ylabel={\scriptsize{$\mathrm{MSE}$}},
]
\addplot [color=mycolor1, line width=1.0pt]
  table[row sep=crcr]{%
-5	0.652677123244914\\
-4	0.603639631403647\\
-3	0.552550426335933\\
-2	0.500096478149767\\
-1	0.447252297917171\\
0	0.395277598009866\\
1	0.345125548091542\\
2	0.297646438631453\\
3	0.253624264456372\\
4	0.213607565446537\\
5	0.178072500625444\\
6	0.146938936137843\\
7	0.120245773640159\\
8	0.097586564579163\\
9	0.0787359707403477\\
10	0.063198927417482\\
11	0.0505210432138252\\
12	0.0402818231694771\\
13	0.0320023340656993\\
14	0.0253941966530881\\
15	0.0201182484771332\\
};
\addlegendentry{RCR: Analytical}

\addplot [only marks, line width =1pt, mark size=1.2pt, mark=o, mark options={solid, red}]
  table[row sep=crcr]{%
-5	0.661436796312411\\
-4	0.607017657978226\\
-3	0.563714469918548\\
-2	0.489651675351948\\
-1	0.451542851494575\\
0	0.417934733963987\\
1	0.357049286821602\\
2	0.299246109076555\\
3	0.260044488643388\\
4	0.21382247180138\\
5	0.168618310629876\\
6	0.144223372896345\\
7	0.123399854325959\\
8	0.097517829210746\\
9	0.0797923451122138\\
10	0.0632680627370651\\
11	0.0496799290604995\\
12	0.041611234641351\\
13	0.0303312654175997\\
14	0.025308025716121\\
15	0.0209822167582728\\
};
\addlegendentry{RCR: Empirical}

\addplot [color=OliveGreen, dashed, line width=1.0pt]
  table[row sep=crcr]{%
-5	0.794643694390757\\
-4	0.748430516118669\\
-3	0.67678154989124\\
-2	0.591447685457408\\
-1	0.549159776920495\\
0	0.486100547920642\\
1	0.43123416571511\\
2	0.378583827896152\\
3	0.328167928123375\\
4	0.26312318400326\\
5	0.224495970734318\\
6	0.190198821228096\\
7	0.160496891708166\\
8	0.130905495479323\\
9	0.106711118723427\\
10	0.0914783105249661\\
11	0.0707287619054264\\
12	0.063523463562782\\
13	0.0468697132631553\\
14	0.036224038447243\\
15	0.0304366040353704\\
};
\addlegendentry{RLS: Empirical}

\end{axis}
\end{tikzpicture}%
\caption{\scriptsize{MSE performance.}}
\label{fig:PSK_MSE}
\end{subfigure}
\begin{subfigure}[h]{.5\textwidth}%
\centering
%
%
%
\definecolor{mycolor1}{rgb}{0.00000,0.44700,0.74100}%
\definecolor{mycolor2}{rgb}{0.85000,0.32500,0.09800}%
\definecolor{OliveGreen}{rgb}{0,0.5,0}%
\begin{tikzpicture}
 \renewcommand{\axisdefaulttryminticks}{4}
    \tikzstyle{every major grid}+=[style=densely dashed]
    \tikzstyle{every axis y label}+=[yshift=-10pt]
    \tikzstyle{every axis x label}+=[yshift=5pt]
    \tikzstyle{every axis legend}+=[cells={anchor=west},fill=white,
        at={(0.01,0.01)}, anchor=south west, font=\tiny ]

\begin{axis}[%
width=2.128in,
height=1.754in,
scale only axis,
xmin=-5,
xmax=15,
xlabel style={font=\color{white!15!black}},
xlabel={\scriptsize{$\rm{SNR \ (dB)}$}},
ymode=log,
ymin=0.01,
ymax=1,
yminorticks=true,
ylabel style={font=\color{white!15!black}},
ylabel={\scriptsize{$\rm{SEP}$}},
axis background/.style={fill=white},
title style={font=\bfseries},
 grid=major,
 scaled ticks=true,
]
\addplot [color=mycolor1, line width=1.0pt]
  table[row sep=crcr]{%
-5	0.7943\\
-4	0.7777\\
-3	0.75798\\
-2	0.7351\\
-1	0.7098\\
0	0.68087\\
1	0.64849\\
2	0.61233\\
3	0.57278\\
4	0.52858\\
5	0.48176\\
6	0.43019\\
7	0.37764\\
8	0.32292\\
9	0.26765\\
10	0.21471\\
11	0.16412\\
12	0.11841\\
13	0.07902\\
14	0.0491\\
15	0.02714\\
};
\addlegendentry{RCR: Analytical}

\addplot [only marks, line width =1pt, mark size=1.2pt, mark=o, mark options={solid, red}]
  table[row sep=crcr]{%
-5	0.80796875\\
-4	0.780625\\
-3	0.7684375\\
-2	0.755625\\
-1	0.7140625\\
0	0.67703125\\
1	0.65234375\\
2	0.61265625\\
3	0.5775\\
4	0.52015625\\
5	0.4890625\\
6	0.44671875\\
7	0.38234375\\
8	0.32265625\\
9	0.28390625\\
10	0.2284375\\
11	0.16921875\\
12	0.11859375\\
13	0.07671875\\
14	0.0509375\\
15	0.028125\\
};
\addlegendentry{RCR: Empirical}

\addplot [color=OliveGreen, dashed, line width=1.0pt]
  table[row sep=crcr]{%
-5	0.858203125\\
-4	0.848046875\\
-3	0.826953125\\
-2	0.80546875\\
-1	0.799609375\\
0	0.774609375\\
1	0.734765625\\
2	0.73671875\\
3	0.7\\
4	0.67421875\\
5	0.633984375\\
6	0.5703125\\
7	0.525\\
8	0.487890625\\
9	0.431640625\\
10	0.390234375\\
11	0.333203125\\
12	0.2796875\\
13	0.23125\\
14	0.16875\\
15	0.125390625\\
};
\addlegendentry{RLS: Empirical}

\end{axis}
\end{tikzpicture}%
\caption{\scriptsize{SEP performance.}}
\label{Fig:PSK_SEP}
\end{subfigure}
\caption{\scriptsize{Performance of the Circular Relaxation (CR) for $16$-PSK as a function of the SNR. The analytical curve is based on Theorem~1.  For the empirical simulations, we used $\kappa =2, n =128$ and data are averaged over $50$ independent MC iterations.}}
\label{fig:PSK MSE/SEP}
\end{figure*}
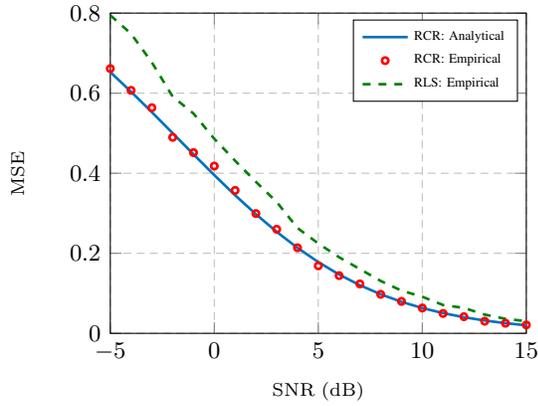
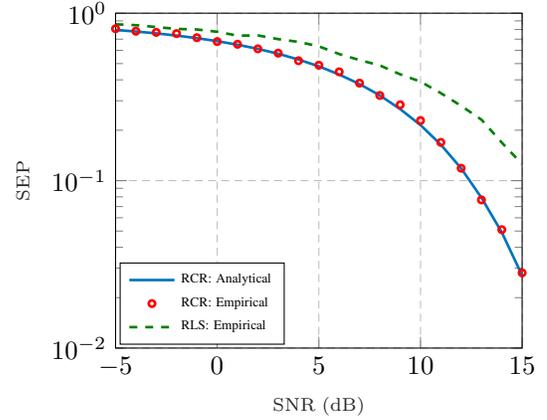
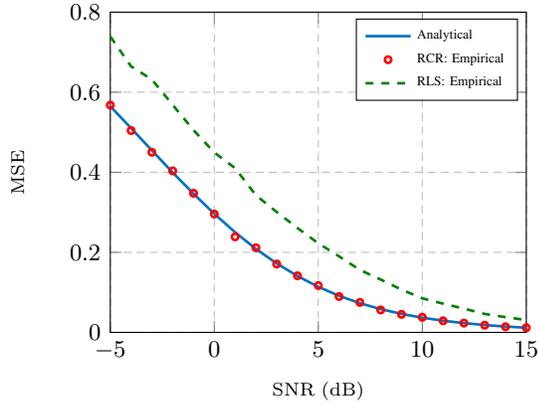
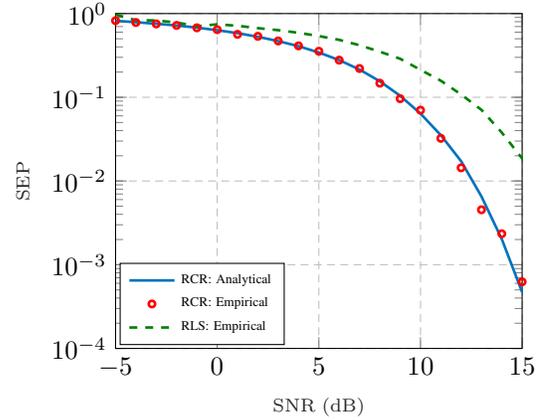
\begin{figure*}
\begin{subfigure}[h]{.5\textwidth}
  \centering
%
%
\definecolor{mycolor1}{rgb}{0.00000,0.44700,0.74100}%
\definecolor{mycolor2}{rgb}{0.85000,0.32500,0.09800}%
\definecolor{OliveGreen}{rgb}{0,0.5,0}%
%
\begin{tikzpicture}[scale=1,font=\small]
    \renewcommand{\axisdefaulttryminticks}{4}
    \tikzstyle{every major grid}+=[style=densely dashed]
    \tikzstyle{every axis y label}+=[yshift=-10pt]
    \tikzstyle{every axis x label}+=[yshift=5pt]
    \tikzstyle{every axis legend}+=[cells={anchor=west},fill=white,
        at={(0.98,0.98)}, anchor=north east, font=\tiny ]
\begin{axis}[%
width=2.8in,
height=2.3in,
xmin=-5,
xmax=15,
ymin=0,
ymax=0.80,
      grid=major,
      scaled ticks=true,
       xlabel={\scriptsize{$\rm{SNR \ (dB)}$}},
       ylabel={\scriptsize{$\mathrm{MSE}$}},
]
\addplot [color=mycolor1, line width=1.0pt]
  table[row sep=crcr]{%
-5	0.564844781908858\\
-4	0.510151459029343\\
-3	0.454549206071825\\
-2	0.399620955794549\\
-1	0.346540609000462\\
0	0.296339305438929\\
1	0.250152152983718\\
2	0.208612165184731\\
3	0.172271167869742\\
4	0.14074304008719\\
5	0.113847753345322\\
6	0.0913614567992313\\
7	0.0729769486029522\\
8	0.058128516655954\\
9	0.0462220174125751\\
10	0.0367001474138383\\
11	0.0291775390147349\\
12	0.0231739359843227\\
13	0.01838926131342\\
14	0.0146315246206278\\
15	0.01160586257321\\
};
\addlegendentry{Analytical}
\addplot [only marks, line width =1pt, mark size=1.2pt, mark=o, mark options={solid, red}]
  table[row sep=crcr]{%
-5	0.567712263567865\\
-4	0.50422681415786\\
-3	0.450105598791514\\
-2	0.403517762062243\\
-1	0.347623783379656\\
0	0.295613596227476\\
1	0.238704482698077\\
2	0.211476285287481\\
3	0.170704938396434\\
4	0.141371922485337\\
5	0.117003672313169\\
6	0.0896032866633698\\
7	0.0748922419932259\\
8	0.056083813855868\\
9	0.0450319577937263\\
10	0.0377422714338062\\
11	0.0289662181946408\\
12	0.0232580007017338\\
13	0.0179457361426157\\
14	0.0141332969299728\\
15	0.0116038427015183\\
};
\addlegendentry{RCR: Empirical}

\addplot [color=OliveGreen, dashed, line width=1.0pt]
  table[row sep=crcr]{%
-5	0.739206924111293\\
-4	0.664563902625815\\
-3	0.631759382485317\\
-2	0.568321701403432\\
-1	0.505503191029522\\
0	0.449377708026567\\
1	0.410406186255206\\
2	0.342559880932105\\
3	0.300805812482958\\
4	0.260683526468875\\
5	0.222604984086232\\
6	0.189896270246394\\
7	0.15602093801618\\
8	0.131983829406063\\
9	0.106773935836679\\
10	0.0851787612171377\\
11	0.071804174166673\\
12	0.0593592420844816\\
13	0.0460435948010212\\
14	0.0380789016062059\\
15	0.0309214742449476\\
};
\addlegendentry{RLS: Empirical}
\end{axis}
\end{tikzpicture}%
\caption{\scriptsize{MSE performance vs. SNR.}}
\label{fig:MSE_QAM}
\end{subfigure}
\begin{subfigure}[h]{.5\textwidth}
\centering
%
%
\definecolor{mycolor1}{rgb}{0.00000,0.44700,0.74100}%
\definecolor{mycolor2}{rgb}{0.85000,0.32500,0.09800}%
\definecolor{OliveGreen}{rgb}{0,0.5,0}%
\begin{tikzpicture}
 \renewcommand{\axisdefaulttryminticks}{4}
    \tikzstyle{every major grid}+=[style=densely dashed]
    \tikzstyle{every axis y label}+=[yshift=-10pt]
    \tikzstyle{every axis x label}+=[yshift=5pt]
    \tikzstyle{every axis legend}+=[cells={anchor=west},fill=white,
        at={(0.01,0.01)}, anchor=south west, font=\tiny ]
\begin{axis}[%
width=2.128in,
height=1.754in,
scale only axis,
xmin=-5,
xmax=15,
xlabel style={font=\color{white!15!black}},
xlabel={\scriptsize{$\rm{SNR \ (dB)}$}},
ymode=log,
ymin=0.0001,
ymax=1,
yminorticks=true,
ylabel style={font=\color{white!15!black}},
ylabel={\scriptsize{$\rm{SEP}$}},
 grid=major,
 scaled ticks=true,
]
\addplot [color=mycolor1, line width=1.0pt]
  table[row sep=crcr]{%
-5	0.81856\\
-4	0.78932\\
-3	0.75547\\
-2	0.72095\\
-1	0.67959\\
0	0.63296\\
1	0.58357\\
2	0.52834\\
3	0.46894\\
4	0.40787\\
5	0.34267\\
6	0.27822\\
7	0.21455\\
8	0.15522\\
9	0.10414\\
10	0.06386\\
11	0.03527\\
12	0.01723\\
13	0.00653\\
14	0.00202\\
15	0.00047\\
};
\addlegendentry{RCR: Analytical}

\addplot [only marks, line width =1pt, mark size=1.2pt, mark=o, mark options={solid, red}]
  table[row sep=crcr]{%
-5	0.8225\\
-4	0.7825\\
-3	0.7528125\\
-2	0.7221875\\
-1	0.67609375\\
0	0.64203125\\
1	0.56484375\\
2	0.53484375\\
3	0.47078125\\
4	0.4115625\\
5	0.355\\
6	0.2778125\\
7	0.220625\\
8	0.14734375\\
9	0.09625\\
10	0.07015625\\
11	0.03234375\\
12	0.014375\\
13	0.00453125\\
14	0.00234375\\
15	0.000625\\
};
\addlegendentry{RCR: Empirical}

\addplot [color=OliveGreen, dashed, line width=1.0pt]
  table[row sep=crcr]{%
-5	0.9725\\
-4	0.8325\\
-3	0.828125\\
-2	0.7821875\\
-1	0.70609375\\
0	0.743828125\\
1	0.715234375\\
2	0.669765625\\
3	0.63359375\\
4	0.5890625\\
5	0.538671875\\
6	0.48546875\\
7	0.41859375\\
8	0.35078125\\
9	0.2884375\\
10	0.2128125\\
11	0.157734375\\
12	0.10796875\\
13	0.070703125\\
14	0.037890625\\
15	0.01859375\\
16	0.006328125\\
17	0.00140625\\
18	0.000703125\\
19	0.000234375\\
20	0\\
21	0\\
22	0\\
23	0\\
24	0\\
25	0\\
};
\addlegendentry{RLS: Empirical}

\end{axis}
\end{tikzpicture}%
\caption{\scriptsize{SEP performance vs. SNR.}}
\label{fig:SER_QAM}
\end{subfigure}
\caption{\scriptsize{Box Relaxation performance for $16$-QAM. The analytical prediction is based on Theorem~1. We used $\kappa =2, n =128$ and data are averaged over $50$ independent MC trials.}}
\label{fig:MSE/SER:QAM}
\end{figure*}
\section{Sketch of the Proof}
\label{Proof}
In this section, we provide a proof sketch of Theorem \ref{Theorem:BER}. For the reader's convenient,  we summarize the main tool of our analysis, namely the CGMT, in the next subsection.
\subsection{Analysis Tool: CGMT}
We first need to state the key ingredient of the analysis which is the CGMT. Here, we just recall the statement of the theorem, and we refer the reader to \cite{thrampoulidis2016precise} for the complete technical requirements.
Consider the following two min-max problems, which we refer to as the Primal Optimization (PO) and the Auxiliary Optimization (AO) problems:
\begin{subequations}
\begin{align}\label{P,AO}
&\Psi(\Cm) := \underset{\av \in \mathcal{S}_{\av}}{\operatorname{\min}}  \ \underset{\bv \in \mathcal{S}_{\bv}}{\operatorname{\max}} \ \bv^{\top} \Cm \av + \mathcal{T}( \av, \bv), \\
&\psi(\gv_1, \gv_2) := \underset{\av \in \mathcal{S}_{\av}}{\operatorname{\min}}  \ \underset{\bv \in \mathcal{S}_{\bv}}{\operatorname{\max}} \ \| \av \| \gv_1^{\top} \bv + \| \bv \| \gv_2^{\top} \av + \mathcal{T}( \av, \bv), \label{AA2}
\end{align}
\end{subequations}
where $\Cm \in \mathbb{R}^{\tilde m \times \tilde n}, \gv_1 \in \mathbb{R}^{\tilde m}, \gv_2 \in \mathbb{R}^{\tilde n}, \mathcal{S}_\av \subset \mathbb{R}^{\tilde n}, \mathcal{S}_\bv \subset \mathbb{R}^{\tilde m}$ and $\mathcal{T}: \mathbb{R}^{\tilde n} \times \mathbb{R}^{\tilde m} \mapsto \mathbb{R}$. Moreover, the function $\mathcal{T}$ is assumed to be independent of the matrix $\Cm$. Denote by $\av_{\Psi} := \av_{\Psi}(\Cm) $, and $\av_{\psi} := \av_{\psi}( \gv_1, \gv_2)$ any optimal minimizers of (\ref{P,AO}) and (\ref{AA2}), respectively. Further let $\mathcal{S}_\av, \mathcal{S}_\bv$ be convex and compact sets, $\mathcal{T}(\av,\bv)$ is convex-concave continuous on $\mathcal{S}_\av \times \mathcal{S}_\bv$, and $\Cm, \gv_1$ and $\gv_2 $ all have i.i.d. standard normal entries.

The equivalence between the PO and AO is formally described in the following theorem, the proof of which can be found in \cite{thrampoulidis2016precise}.
\begin{theorem}
Let $\mathcal{S}$ be any arbitrary open subset of $\mathcal{S}_\av $, and $\mathcal{S}^c = \mathcal{S}_\av \setminus\mathcal{S}$. Denote $\psi_{\mathcal{S}^c}(\gv_1,\gv_2)$ the optimal cost of the optimization in (\ref{AA2}), when the minimization over $\av$ is constrained over $\av \in \mathcal{S}^c$. 
 Suppose that there exist constants $\eta < \delta$,  such that 
$\psi(\gv_1,\gv_2) \overset{P}{\longrightarrow} \eta$,  and 
 $\psi_{\mathcal{S}^c}(\gv_1,\gv_2) \overset{P}{\longrightarrow} \delta$.\\
Then, $\lim_{\tilde n \rightarrow \infty} \mathbb{P}[\av_{\psi} \in \mathcal{S}] = 1,$ and 
$\lim_{\tilde n \rightarrow \infty} \mathbb{P}[\av_{\Psi} \in \mathcal{S}] = 1.$
\end{theorem}
\subsection{Asymptotic Analysis}
In this part, we provide an outline of the ideas used to prove our main results based on the CGMT framework.
We start by rewriting \eqref{a} by a change of variable to the error vector $\wv := \sv - \sv_0$, to get:
\begin{align}\label{w1}
\hat{\wv}:= \arg \min_{\wv \in \mathcal{V}^n -\sv_0} \frac{1}{2}\|  \Hm \wv-\vv \|^2 + \frac{\zeta}{2} \| \wv + \sv_0 \|^2.
\end{align}
Next, let $\widetilde{\Hm}\! \!=\!\! \!  \left[\begin{matrix}
  \Hm_{\mathsmaller{R}} & -\Hm_{\mathsmaller{I}} \\
  \Hm_{\mathsmaller{I}} & \Hm_{\mathsmaller{R}}
  \end{matrix} \!\right]\! \!\!\in\! \mathbb{R}^{2m \times 2n}$, and $\widetilde{\vv}  \! \!=\! \!  \left[\begin{matrix}
  \vv_{\mathsmaller{R}} \\
  \vv_{\mathsmaller{I}} 
  \end{matrix}\!\right]\! \! \!\in\! \mathbb{R}^{2m}$, where $\Hm_{\mathsmaller{R}}, \vv_{\mathsmaller{R}}$ ($\Hm_{\mathsmaller{I}}, \vv_{\mathsmaller{I}}$) are the real (imaginary) parts of $\Hm$ and $\vv$, respectively. With this, and normalizing \eqref{w1} by $1/n$ we get
  \begin{align}\label{w11}
\widehat{\wv}:= \argmin_{\substack{\widetilde\wv \in \mathbb{R}^{2n}\\ \widetilde{w}_i + j \widetilde{w}_{i+n} \in \mathcal{V}-s_{0,i}}} \frac{1}{2n}\| \frac{1}{\sqrt{2n}} \widetilde\Hm \widetilde\wv-\widetilde\vv  \|^2 +  \frac{\zeta}{2n}\| \widetilde\wv + \widetilde\sv_0 \|^2,
\end{align}
where $\widetilde{\sv}_0 \!=\!  \left[\begin{matrix}
  \sv_{0,{\mathsmaller{R}}} \\
  \sv_{0,{\mathsmaller{I}}} 
  \end{matrix}\right]\!\! \!\in\! \mathbb{R}^{2n}$.
Because of the dependence between the entries of $\widetilde{\Hm}$, 
the above optimization is difficult to analyze and the CGMT framework cannot be used directly here. However, as discussed in \cite{abbasi2019performance}, one can use Lindeberg methods as in \cite{oymak2018universality} to replace $\widetilde{\Hm}$ with a Gaussian matrix that has i.i.d. entries without affecting the asymptotic performance, then we get
 \begin{align}\label{w2}
\widehat{\wv}= \argmin_{\substack{\widetilde\wv \in \mathbb{R}^{2n}\\ \widetilde{w}_i + j \widetilde{w}_{i+n} \in \mathcal{V}-s_{0,i}}} \frac{1}{2n}\| \frac{1}{\sqrt{2n}}\Am \widetilde\wv -  \widetilde\vv \|^2 +  \frac{\zeta}{2n}\| \widetilde\wv + \widetilde\sv_0 \|^2,
\end{align}
\noindent
where $\Am \in \mathbb{R}^{2m \times 2n}$ has i.i.d. $\mathcal{N}(0,1)$ entries, and $\widetilde \vv$ has i.i.d. $\mathcal{N}(0,\frac{\sigma^2}{2})$ elements.
Next, we proceed to apply the CGMT by rewriting \eqref{w2} as the following min-max optimization:
\begin{align}\label{PO}
\min_{\substack{\widetilde\wv \in \mathbb{R}^{2n}\\ \widetilde{w}_i + j \widetilde{w}_{i+n} \in \mathcal{V}-s_{0,i}}}\! \! \max_{\uv \in \mathbb{R}^{2m}} \  \frac{\uv^{\top} \Am \widetilde \wv}{2n \sqrt{2n}} -\! \frac{\uv^{\top} \widetilde \vv}{2n} \!-\! \frac{\| \uv\|^2}{8n} \!+\!  \frac{\zeta}{2n}\| \widetilde\wv \!+\! \widetilde\sv_0 \|^2.
\end{align}
One technical requirement of the CGMT is the compactness of the feasibility set over $\uv$. This can be handled according to the approach in \cite[Appendix A]{thrampoulidis2016precise}, by introducing a sufficiently large {artificial} constraint set $\mathcal{S}_\uv$ which will not affect the optimization problem with high probability to obtain:
\begin{align}\label{PO2}
\min_{\substack{\widetilde\wv \in \mathbb{R}^{2n}\\  \widetilde{w}_i + j \widetilde{w}_{i+n} \in \mathcal{V}- s_{0,i}}}\! \! \max_{\uv \in \mathcal{S}_\uv} \ \frac{\uv^{\top} \Am \widetilde \wv}{2n \sqrt{2n}} -\! \frac{\uv^{\top} \widetilde \vv}{2n} \!-\! \frac{\| \uv\|^2}{8n} \!+\!  \frac{\zeta}{2n}\| \widetilde\wv \!+\! \widetilde\sv_0 \|^2.
\end{align}
The above problem is in the form of a PO of the CGMT, hence we can associate to it the following simplified AO optimization problem:
\begin{align}\label{AO}
\min_{\substack{\widetilde\wv \in \mathbb{R}^{2n}\\ \widetilde{w}_i + j \widetilde{w}_{i+n} \in \mathcal{V}-s_{0,i}}}& \max_{\uv \in \mathcal{S}_\uv} \   \frac{ \|\widetilde\wv\| \gv_1^{\top} \uv}{2n \sqrt{2n}} + \frac{ \|\uv\| \gv_2^{\top} \widetilde\wv}{2n \sqrt{2n}}
\nonumber\\
&- \frac{\uv^{\top} \widetilde \vv}{2n}  - \frac{\| \uv\|^2}{8n}+  \frac{\zeta}{2n}\| \widetilde\wv + \widetilde\sv_0 \|^2,
\end{align}
where $\gv_1 \in \mathbb{R}^{2m}$ and $\gv_2 \in \mathbb{R}^{2n}$ have i.i.d. $\mathcal{N}(0,1)$ entries.
With some abuse of notation on $\gv_1$, we can see that $\left(\frac{\|\widetilde \wv\|}{\sqrt{2n}} \gv_1 - \widetilde \vv \right)^{\top} \uv \ {\overset{d} =} \  \gv_1^{\top} \uv \ \sqrt{\frac{\|\widetilde \wv\|^2}{{2n}} +\frac{\sigma^2}{2}} $. Hence, \eqref{AO} becomes
\begin{align}\label{AO1}
\min_{\substack{\widetilde\wv \in \mathbb{R}^{2n}\\ \widetilde{w}_i + j \widetilde{w}_{i+n}  \in \mathcal{V}-s_{0,i}}}& \max_{\uv \in \mathcal{S}_\uv} \gv_1^{\top} \uv \sqrt{\frac{\|\widetilde \wv\|^2}{{2n}} +\frac{\sigma^2}{2}} + \frac{ \|\uv\| \gv_2^{\top} \widetilde\wv}{2n \sqrt{2n}} \nonumber\\
&- \frac{\| \uv\|^2}{8n}+  \frac{\zeta}{2n}\| \widetilde\wv + \widetilde\sv_0 \|^2.
\end{align}
Fixing $\beta: = \frac{\|\uv\|}{\sqrt{2n}}$, the optimization over $\uv$ simplifies to
\begin{align}\label{AO2}
\min_{\substack{\widetilde\wv \in \mathbb{R}^{2n}\\ \widetilde{w}_i + j \widetilde{w}_{i+n}  \in \mathcal{V} - s_{0,i}}} &\max_{\beta>0} \ \frac{\beta \|\gv_1\|}{\sqrt{2n}} \sqrt{\frac{\|\widetilde \wv\|^2}{{2n}} +\frac{\sigma^2}{2}} \nonumber\\
&+ \frac{ \beta\gv_2^{\top} \widetilde\wv}{2n} - \frac{\beta^2}{4}+  \frac{\zeta}{2n}\| \widetilde\wv + \widetilde\sv_0 \|^2.
\end{align}
The square root in the above problem can be expressed using the following identity\footnote{Note that at optimality, $ \tau_* = {\chi}$.}
\begin{align}
{\chi} = \min_{\tau>0}\frac{1}{2} \left( \frac{\chi^2}{ \tau} + {\tau}  \right), \quad \quad  \text{for} \ \chi>0,
\end{align}
which yields the following optimization problem
\begin{align}\label{AO3}
&\min_{\tau>0} \  \max_{\beta>0} \ \frac{\tau \beta \|\gv_1\|}{2\sqrt{2n}} +\frac{\sigma^2\|\gv_1\| \beta}{4 \tau \sqrt{2n}} - \frac{\beta^2}{4}\nonumber\\
&+\min_{\substack{\widetilde\wv \in \mathbb{R}^{2n}\\ \widetilde{w}_i + j \widetilde{w}_{i+n}  \in \mathcal{V} - s_{0,i}}}  \frac{ \beta \|\gv_1\|}{2\tau \sqrt{2n}} \frac{\|\widetilde \wv\|^2}{2n}+\frac{ \beta\gv_2^{\top} \widetilde\wv}{2n} + \frac{\zeta}{2n}\| \widetilde\wv + \widetilde\sv_0 \|^2.
\end{align}
Using the weak law of large numbers (WLLN): $\frac{\|\gv_1\|}{\sqrt{2n}} \pto \sqrt{\kappa}$, then the above problem reduces to
\begin{align}\label{AO4}
&\min_{\tau>0} \  \max_{\beta>0}  \ \ \frac{\tau \beta \sqrt{\kappa}}{2} +\frac{\sigma^2 \beta \sqrt{\kappa}}{4 \tau} - \frac{\beta^2}{4}
\nonumber\\
&+\min_{\substack{\widetilde\wv \in \mathbb{R}^{2n}\\  \widetilde{w}_i + j \widetilde{w}_{i+n} \in \mathcal{V} - s_{0,i}}}  \frac{ \beta \sqrt{\kappa}}{2\tau } \frac{\|\widetilde \wv\|^2}{2n}+\frac{ \beta\gv_2^{\top} \widetilde\wv}{2n} + \frac{\zeta}{2n}\| \widetilde\wv + \widetilde\sv_0 \|^2.
\end{align}
Defining $\alpha:= \frac{\tau}{\sqrt{\kappa}}$, and by a completion of squares in the minimization over $\widetilde \wv$, and using the WLLN, we obtain the following scalar (deterministic) optimization problem
\begin{align}\label{AO5}
\min_{\alpha>0}& \max_{\beta>0} \frac{\alpha \beta {\kappa}}{2} +\frac{\sigma^2 \beta }{4 \alpha} - \frac{\beta^2}{4}
- \frac{\beta^2}{\frac{2 \beta}{\alpha} +4 \zeta} +  \frac{1}{2} \left( \frac{\beta}{2 \alpha} - \frac{\frac{\beta^2}{4 \alpha^2}}{\frac{\beta}{2\alpha}+\zeta} \right)\nonumber \\
& +\frac{1}{2} \left( \frac{\beta}{2 \alpha} +\zeta \right) \mathbb{E}\left[ {\mathcal{D}}^2\left(\frac{\frac{\beta}{2 \alpha}}{\frac{\beta}{2 \alpha}+\zeta} S_0 - \frac{\beta}{\frac{\beta}{ \alpha}+2\zeta} G_c; \mathcal{V}\right) \right].
\end{align}
The $\rm SEP$ of $\widehat{\wv}$ in \eqref{w2} can be derived in a similar way to the proof of \cite{abbasi2019performance} to get
\begin{equation}\label{sepf}
{\rm SEP} \!\pto\! \mathbb{P} \!\left[ {\mathbf \Pi} \! \left(\frac{\frac{\beta_*}{2 \alpha_*}}{\frac{\beta_*}{2 \alpha_*} \!+ \!\zeta} S_0 \!-\! \frac{\beta_*}{\frac{\beta_*}{ \alpha_*}+2\zeta} G_c;\! \mathcal{V} \! \right) \!\notin\! \mathcal{V}_{S_0} \! \right]\!,
\end{equation}
simplifying \eqref{AO5} and \eqref{sepf} concludes the proof of the SEP part of Theorem~1.

The MSE expression can be proven in a similar way by noting that 
\begin{align}
\frac{\|\widetilde \wv \|^2}{2n}+ \frac{\sigma^2}{2} = \hat{\tau}_n^2,
\end{align}
where $\hat{\tau}_n$ is the solution to \eqref{AO4}. Hence, using $\hat\alpha_n = \frac{\hat\tau_n}{\sqrt{\kappa}}$, and $\hat\alpha_n \pto \alpha_*$, where $\alpha_*$ is the solution of \eqref{AO5}, we conclude, by applying the CGMT, that
\begin{align}
\frac{\|\hat\wv \|^2}{n} \pto 2 \kappa \alpha_*^2 -\sigma^2,
\end{align}
which completes the proof of the MSE part of Theorem~1.
\section{Conclusion}
In this article, we provided sharp performance analysis of the regularized convex relaxation detector when used in complex-valued data detection. 
In particular, we studied its MSE and SEP performance in a massive MIMO application with arbitrary constellation schemes such as QAM and PSK.
Numerical simulations show a close match to the obtained asymptotic results. 
In addition, the derived results can be used to optimally select the detector's parameters such as the regularization factor.
Furthermore, we showed that this convex relaxation outperforms the unconstrained RLS.
\bibliographystyle{IEEEbib}
\bibliography{References}

\begin{thebibliography}{10}

\bibitem{ngo2013energy}
Hien~Quoc Ngo, Erik~G Larsson, and Thomas~L Marzetta,
\newblock ``Energy and spectral efficiency of very large multiuser mimo
  systems,''
\newblock {\em IEEE Transactions on Communications}, vol. 61, no. 4, pp.
  1436--1449, 2013.

\bibitem{verdu1998multiuser}
Sergio Verdu et~al.,
\newblock {\em Multiuser detection},
\newblock Cambridge university press, 1998.

\bibitem{thrampoulidis2016precise}
Christos Thrampoulidis, Ehsan Abbasi, and Babak Hassibi,
\newblock ``Precise error analysis of regularized $ m $-estimators in high
  dimensions,''
\newblock {\em IEEE Transactions on Information Theory}, vol. 64, no. 8, pp.
  5592--5628, 2018.

\bibitem{stojnic2013framework}
Mihailo Stojnic,
\newblock ``A framework to characterize performance of lasso algorithms,''
\newblock {\em arXiv preprint arXiv:1303.7291}, 2013.

\bibitem{thrampoulidis2018symbol}
Christos Thrampoulidis, Weiyu Xu, and Babak Hassibi,
\newblock ``Symbol error rate performance of box-relaxation decoders in massive
  mimo,''
\newblock {\em IEEE Transactions on Signal Processing}, vol. 66, no. 13, pp.
  3377--3392, 2018.

\bibitem{thrampoulidis2015regularized}
Christos Thrampoulidis, Samet Oymak, and Babak Hassibi,
\newblock ``Regularized linear regression: A precise analysis of the estimation
  error,''
\newblock in {\em Conference on Learning Theory}. PMLR, 2015, pp. 1683--1709.

\bibitem{atitallah2017box}
Ismail~Ben Atitallah, Christos Thrampoulidis, Abla Kammoun, Tareq~Y
  Al-Naffouri, Mohamed-Slim Alouini, and Babak Hassibi,
\newblock ``The box-lasso with application to gssk modulation in massive mimo
  systems,''
\newblock in {\em 2017 IEEE International Symposium on Information Theory
  (ISIT)}. IEEE, 2017, pp. 1082--1086.

\bibitem{atitallah2017ber}
Ismail~Ben Atitallah, Christos Thrampoulidis, Abla Kammoun, Tareq~Y
  Al-Naffouri, Babak Hassibi, and Mohamed-Slim Alouini,
\newblock ``Ber analysis of regularized least squares for bpsk recovery,''
\newblock in {\em 2017 IEEE International Conference on Acoustics, Speech and
  Signal Processing (ICASSP)}. IEEE, 2017, pp. 4262--4266.

\bibitem{alrashdi2020optimum}
Ayed~M Alrashdi, Abla Kammoun, Ali~H Muqaibel, and Tareq~Y Al-Naffouri,
\newblock ``Optimum m-pam transmission for massive mimo systems with channel
  uncertainty,''
\newblock {\em arXiv preprint arXiv:2008.06993}, 2020.

\bibitem{alrashdi2019precise}
Ayed~M Alrashdi, Ismail~Ben Atitallah, and Tareq~Y Al-Naffouri,
\newblock ``Precise performance analysis of the box-elastic net under matrix
  uncertainties,''
\newblock {\em IEEE Signal Processing Letters}, vol. 26, no. 5, pp. 655--659,
  2019.

\bibitem{abbasi2019performance}
Ehsan Abbasi, Fariborz Salehi, and Babak Hassibi,
\newblock ``Performance analysis of convex data detection in mimo,''
\newblock in {\em ICASSP 2019-2019 IEEE International Conference on Acoustics,
  Speech and Signal Processing (ICASSP)}. IEEE, 2019, pp. 4554--4558.

\bibitem{oymak2018universality}
Samet Oymak and Joel~A Tropp,
\newblock ``Universality laws for randomized dimension reduction, with
  applications,''
\newblock {\em Information and Inference: A Journal of the IMA}, vol. 7, no. 3,
  pp. 337--446, 2018.

\end{thebibliography}
\end{document}